\documentclass[final,copyright,creativecommons]{eptcs}
\providecommand{\event}{EXPRESS/SOS 2013}

\usepackage{pdfsync}
\usepackage{stmaryrd}
\usepackage{amsmath,amssymb,amsthm}
\usepackage{mathrsfs}
\usepackage{extarrows}
\usepackage{xspace}
\usepackage{tikz}
	\usetikzlibrary{arrows,shapes,automata,backgrounds,decorations}
\usepackage{graphicx}
\usepackage{mdwlist,paralist}
\usepackage{url}
\usepackage{xifthen}

\usepackage{PriorityInEventStructure}

\title{Adding Priority to Event Structures%
	\thanks{Supported by the DFG Research Training Group SOAMED.}}
\author{Youssef Arbach\thanks{Corresponding author: \url{arbach@soamed.de}} 
  \qquad\qquad Kirstin Peters 
  \qquad\qquad Uwe Nestmann 
  \\
  \institute{Technische Universit\"at Berlin, Germany}
}

\def\titlerunning{Adding Priority to Event Structures}
\def\authorrunning{Y. Arbach, K. Peters, U. Nestmann}

\begin{document}

\maketitle


\begin{abstract}
  Event Structures (ESs) are mainly concerned with the representation of
  causal relationships between events, usually accompanied by other event
  relations capturing conflicts and disabling.    
  Among the most prominent variants of ESs are \emph{Prime} ESs,
  \emph{Bundle} ESs, \emph{Stable} ESs, and \emph{Dual} ESs, which differ in their causality models and event relations.
  Yet, some application domains require further kinds of relations between
  events.  Here, we add the possibility to express priority relationships
  among events.

  We exemplify our approach on Prime, Bundle, Extended Bundle, and Dual ESs.
  Technically, we enhance these variants in the same way.
  For each variant, we then study the interference between priority and the
  other event relations.
  From this, we extract the redundant priority pairs|notably differing for the types of ESs|that enable us to provide a comparison between the extensions.  
  We also exhibit that priority considerably complicates the definition of
  partial orders in ESs.
\end{abstract}


\section{Introduction}

\paragraph{Concurrency Model.}
\label{sec:concurrency-mod}

Event Structures (ESs) are concerned with usually statically defined
relationships that govern the possible occurrences of events, typically
represented as \emph{causality} (for precedence) and \emph{conflict} (for
choice). An event is a single occurrence or an instance of an action, and
thus cannot be repeated.
ESs were first used to give semantics to
Petri nets in \cite{Winskel:Thesis}, then to give semantics
to process calculi in \cite{ESComparison,Langerak:Thesis}, and concurrent systems in general in \cite{EquivalenceNotions}.
The dynamics of an ES are usually provided either by the sets of
traces compatible with the constraints, or by means of configurations based
sets of events, possibly in their labeled partially-ordered variant
(\emph{lposets}).

Event Structures are \emph{non-interleaving} models. In interleaving
models, events take place linearly, one after the other.  There,
concurrency is expressed in terms of free choice or non-determinism, \ie
concurrent events can appear in any order.  Event Structures show
concurrency not as a linear order with free choice, but as independence
between events, \ie events are concurrent when they are related neither by
conflict nor causally.  This intuition manifests clearly in system
runs, represented by the so-called configurations, or in terms of partial
orders, where concurrent events are unordered.

\vspace{-0.5em}

\paragraph{Application Domain.}
\label{sec:application-domain}

We investigate the phenomenon of Dynamic Coalitions (DC). The term denotes
a temporary collaboration between entities in order to achieve a common
goal. Afterwards, such a coalition resolves itself, or is resolved.  One
example of a DC is the treatment of a stroke patient, taken from the
medical sector, which inspires our work:
\emph{A patient gets a stroke which calls for the ambulance. In
the meanwhile, the emergency room prepares to receive the patient. Then,
the ambulance arrives and the patient is transferred to the emergency
room. While the patient is in the emergency room, the latter
communicates with the stroke unit to prepare for transferring the patient,
and then the patient is sent to the stroke unit. Before the patient is
discharged from the stroke unit, some therapists are invited by the stroke
unit to join the patient treatment.}
Such coalitions are called dynamic, as they evolve over time, where new
members can join, and others can leave, until the goal is achieved. We call
this specific phenomenon the \emph{formation} of a DC. Others call it the
\emph{membership dimension} of a DC \cite{Bryans0601}.

Examining the application scenario, we observe that it can be naturally
modeled by means of Event Structures.  Firstly, it mentions events, \eg a
patient gets a stroke, the ambulance joins, and the stroke unit invites
some therapists.  Moreover, we are dealing with possible conflicts between
the members, \eg between the therapists. In addition, there is causality,
for example the event, where the patient gets a stroke, causes the
ambulance to join and the emergency room to prepare. Finally, there can be
concurrency, \ie multiple members of the coalition can work concurrently,
\eg the stroke unit prepares to receive the patient while the emergency
room is still working on the patient.

\vspace{-0.5em}

\paragraph{Further Requirements.}
\label{sec:requirements}

Applications may impose to limit the amount of concurrency in the
specifications of some systems or, in interleaving models, to limit the
non-determinism or free choice. For example, in our above-mentioned
healthcare scenario, imagine that while the therapists are working on the
patient outside the hospital, the patient gets another stroke, and then the ambulance again
needs to involve and interrupt the work of the therapists. So, they
cannot all work together at the same time. Besides, in this particular
situation, the ambulance should have a higher \emph{priority} to perform its
events, such that only afterwards the therapists might continue their work.
This is some kind of order, so there is a determinism here on who needs to
go first, carried by the concept of priority. 
The precedence caused by priority is called ``pre-emption'' (cf.\
\cite{Luettgen:Thesis}). So, the event with higher priority pre-empts the
event with lower priority: the higher-priority one must happen before any
concurrently enabled event with lower priority, so a priority relation is
(only) applied in a state of competition. For example, some processes
compete to run on a processor. In the same way many members of a 
coalition compete for the patient. Some of them can work concurrently, due
to the specifications of the system, and some cannot. 

\vspace{-0.5em}

\paragraph*{Overview.}
\label{sec:overview}

This paper is organized as follows. In \S\ref{sec:PPES}, we start with the simplest form of ESs, the Prime ESs, add priority to it, discuss the overlapping between priority and the other relations of Prime ESs, and show how to reduce this overlapping.
In \S\ref{sec:PBES}, we introduce Bundle ESs, their traces, configurations, and lposets. Then we add priority to them and investigate the relation between priority and other event relations of BESs like enabling and precedence.
In \S\ref{sec:PEBES} and \S\ref{sec:PDES} we then study the two extensions Extended Bundle ESs and Dual ESs of Bundle ESs and how their different causality models modify the relationship of priority and the other event relations of the ESs.
In \S\ref{sec:conclusions}, we summarize the work and conclude by comparing the
results.

\vspace{-0.5em}

\paragraph{Related Work.}
\label{sec:related-work}

Priorities are used as a mechanism to provide interrupts in systems concerned with processes. For example, in Process Algebra, Baeten et al.\ were the first to add priority in \cite{Baeten:1986}.  They defined it as a partial order relation \(<\). Moreover, Camilleri, and Winskel integrated priority within the Summation operator in CCS \cite{Camilleri:CCSPriority}.
Also, L\"uttgen in \cite{Luettgen:Thesis} and Cleaveland et al.\ in \cite{PriorityInProcessAlgebra} considered the semantics of Process Algebra with priority.

In Petri Nets, which are a non-interleaving model like Event Structures, Bause in \cite{PetriNetsWithStaticPriority,PetriNetsWithDynamicPriority} added priority to Petri Nets in two different ways: static and dynamic. Dynamic means the priority relation evolves during the system run, while the static one means it is fixed since the beginning and will never change till the end. In that sense, static priority is what we define here.

In this paper, we add priority to different kinds of Event Structures.
Then we analyze the overlapping between the priority relation and the other relations of the respective Event Structure in order to identify and remove redundant priority pairs. We observe that the possibility to identify and remove redundant priority pairs is strongly related to the causality model that is provided by the considered model of ESs.


\section{Priority in Prime Event Structures}
\label{sec:PPES}

Prime Event Structures (PESs), invented by Winskel~\cite{Winskel:Thesis},
are the simplest and first version of ESs. Causality is expressed in terms
of an enabling relation, \ie a partial order between events. For an event
to become enabled in PESs, all of its predecessors with respect to the
enabling relation must take place; an enabled event may happen, but does
not have to do so. There is also a conflict relation between events to provide
choices, given as a binary symmetric relation, and a labeling function
mapping events to actions.

\begin{mydef}
	A \emph{Prime Event Structure (PES)} is a quadruple $ \delta = \left( E, \#, \leq, \labeling{} \right) $ where:
	\begin{compactitem}
		\item \(E\), a set of \emph{events}
		\item \(\# \, \subseteq E \times E\), an irreflexive symmetric relation (the \emph{conflict} relation)
		\item \(\leq \, \subseteq E \times E\), a partial order (the \emph{enabling} relation)
		\item \(\labeling{} : E \to Act\), a \emph{labeling} function
	\end{compactitem}
	that additionally satisfies the following constraints:
	\begin{center}
		\vspace{-0.3em}
		\begin{tabular}{lll}
			1. & \textbf{Conflict Heredity:} \quad \quad & $\forall e, e', e'' \in E \logdot e \# e' \wedge e' \leq e'' \implies e \# e''$\\
			2. & \textbf{Finite Causes:} & $\forall e \in E \logdot \Set{ e' \in E \mid e' \leq e }$ is finite
		\end{tabular}
	\end{center}
\end{mydef}

\begin{figure}[tb]
	\centering
	\vspace{-1em}
	\begin{tikzpicture}[bend angle=45]
		\node (a1)	at (0, 1)		{};
		\node (b1)	at (0.8, 1)		{};
		\node (c1)	at (1.6, 1)		{};
		\node (d1)	at (0.8, 0)		{};
		\node (e1)	at (0, 2)		{};
		\node (a)	at (0.8, -0.8)	{(a)};
		\fill (a1) circle (2.5pt) node [left]		{$a$};
		\fill (b1) circle (2.5pt) node [above right]	{$b$};
		\fill (c1) circle (2.5pt) node [right]		{$c$};
		\fill (d1) circle (2.5pt) node [below]		{$d$};
		\fill (e1) circle (2.5pt) node [right]		{$e$};
		\draw\enabling{2.5}{2.5}	(a1) -- (d1);
		\draw[enabling]			(b1) -- (d1);
		\draw[enabling]			(e1) -- (a1);
		\draw[conflict]			(b1) -- (c1);
		\draw\conflict{2.5}{2.5}	(c1) -- (d1);
		\node (a2)	at (3, 1)		{};
		\node (b2)	at (3.8, 1)		{};
		\node (c2)	at (4.6, 1)		{};
		\node (d2)	at (3.8, 0)		{};
		\node (e2)	at (3, 2)		{};
		\node (b)	at (3.8, -0.8)	{(b)};
		\fill (a2) circle (2.5pt) node [left]		{$a$};
		\fill (b2) circle (2.5pt) node [above right]	{$b$};
		\fill (c2) circle (2.5pt) node [right]		{$c$};
		\fill (d2) circle (2.5pt) node [below]		{$d$};
		\fill (e2) circle (2.5pt) node [right]		{$e$};
		\draw\enabling{2.5}{2.5}	(a2) -- (d2);
		\draw[enabling]			(b2) -- (d2);
		\draw[enabling]			(e2) -- (a2);
		\draw[conflict]			(b2) -- (c2);
		\draw\conflict{2.5}{2.5}	(c2) -- (d2);
		\draw[priority] (a2) .. controls (3.2, 1.7) and (3.6, 1.7) .. (b2);
		\draw[priority] (c2) to [bend left]	(d2);
		\draw[priority] (e2) .. controls (2.6, 2) and (1.8, 0.6) .. (d2);
		\node (a3)	at (6, 1)		{};
		\node (b3)	at (6.8, 1)		{};
		\node (c3)	at (7.6, 1)		{};
		\node (d3)	at (6.8, 0)		{};
		\node (e3)	at (6, 2)		{};
		\node (c)	at (6.8, -0.8)	{(c)};
		\fill (a3) circle (2.5pt) node [left]		{$a$};
		\fill (b3) circle (2.5pt) node [above right]	{$b$};
		\fill (c3) circle (2.5pt) node [right]		{$c$};
		\fill (d3) circle (2.5pt) node [below]		{$d$};
		\fill (e3) circle (2.5pt) node [right]		{$e$};
		\draw\enabling{2.5}{2.5}	(a3) -- (d3);
		\draw[enabling]			(b3) -- (d3);
		\draw[enabling]			(e3) -- (a3);
		\draw[conflict]			(b3) -- (c3);
		\draw\conflict{2.5}{2.5}	(c3) -- (d3);
		\draw[priority] (a3) .. controls (6.2, 1.7) and (6.6, 1.7) .. (b3);
	\end{tikzpicture}
	\vspace{-1em}
	\caption{A Prime ES without priority in (a), with priority in (b), and after dropping redundant priority pairs in (c).}
	\label{fig:PESExamples}
\end{figure}
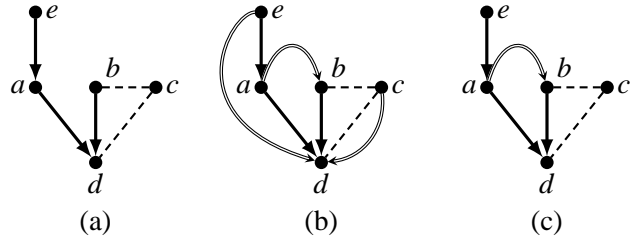

\noindent
Figure~\ref{fig:PESExamples}~(a) shows an example of a Prime ES, where a single-line arrow represents enabling, directed from predecessors in $\leq$ to successors. The dashed line represents a conflict. Note that this structure fulfills the two constraints of a PES.

A \emph{configuration} is a representation of system state by means of the set of events that have occurred up to a certain point. In Prime ESs, a configuration is a conflict-free set of events $ C \subseteq E $ that is left-closed under the enabling relation, \ie no two events of $C$ are in conflict and for all predecessors $e$ with respect to $\leq$ of an event $e' \in C$ it holds $e \in C$. Thus, given a Prime ES $\delta=\left(E, \#, \leq, \labeling{} \right)$, a configuration $C$ represents a \emph{system run} of $ \delta $ (or the \emph{state} of $ \delta $ after this run), where events not related by $\leq$ occur concurrently. 

A \emph{trace} is a sequential version of a system run. It can be defined as a sequence of events which are conflict-free and where all the predecessors of an event in $\leq$ precede that event in the trace. We will define it formally in another equivalent way, which we will rely on when we define priority later:

Let \(\sigma\) be a sequence of events \(e_{1} \dotsm e_{n}\) such that $
\Set{e_{1}, \ldots, e_{n} } \subseteq E $ in a PES $ \delta = \left( E,
  \#, \leq, \labeling{} \right) $. We refer to $\Set{ e_{1}, \ldots, e_{n}
} $ by $\bar{\sigma}$, and we call $\en{\delta}{\sigma}$ the set of events
that are enabled by \(\sigma\), where
\vspace{-0.3em}
\begin{equation}
	\label{eq:PrimeEnablingDef}
	\en{\delta}{\sigma} \deff \Set{ e \in \left( E \setminus \bar{\sigma} \right) \mid \left( \forall e' \in E \logdot e' \leq e \implies e' \in \bar{\sigma} \right) \wedge \left( \nexists e' \in \bar{\sigma} \logdot e \# e' \right) }.
	\vspace{-0.3em}
\end{equation}
We use $\sigma_i$ to denote the prefix $e_1 \dotsm e_i$, for some $i<n$. Then, the sequence $\sigma = e_{1} \dotsm e_{n}$ is called a \emph{trace of \(\delta\)} iff
\vspace{-0.3em}
\begin{equation}
	\label{eq:PrimeTraceDef}
	\forall i \leq n \logdot e_{i} \in \en{\delta}{\sigma_{i-1}}
	\vspace{-0.3em}
\end{equation}

Accordingly, a trace is a linearization of a configuration respecting $\leq$. Usually many traces can be derived from one configuration. The differences between such traces of the same configuration result from concurrent events that are independent, \ie are related neither by enabling nor conflict.
For example, in Figure~\ref{fig:PESExamples}~(a), the events $c$ and $a$ are independent and thus concurrent in a configuration like $\Set{ e, a, c }$. From $\Set{ e, a, c }$ the traces $e a c$, $e c a$, and $c e a$ can be derived for the structure in Figure~\ref{fig:PESExamples}~(a).

If we add priority to PESs, it should be a binary relation between events 
such that, whenever two concurrent events ordered in priority are enabled together, the one with the higher priority must pre-empt the other.\footnote{In fact we could define it as a partial order. However after dropping redundant priority pairs as explained later the priority relation is usually no longer transitive, \ie no longer a partial order.}
Thus we add a new acyclic relation $\lessdot \subseteq E \times E$, the \emph{priority} relation, to Prime ESs and denote the pair $ \left( \delta, \lessdot \right) $ as \emph{prioritized Prime ES (PPES)}. Later on, we add priority in a similar way to other kinds of Event Structures. Sometimes, we expand a prioritized ES $ \left( \delta, \lessdot \right) $, where $ \delta = \left( E, r_1, r_2, \labeling{} \right) $, to $\left( E, r_1, r_2, \labeling{}, \lessdot \right)$.

Figure~\ref{fig:PESExamples}~(b) illustrates a prioritized variant of Figure~\ref{fig:PESExamples}~(a), where the priority relation is represented by a double-lined arrow from the higher-priority event to the lower-priority one, showing the direction of precedence (pre-emption). Sometimes representing both an Event Structure and its associated priority relation in the same diagram becomes too confusing. In that case we visualize the priority relation in a separate diagram next to the structure.

Let us define the interpretation of $\lessdot$ in a formal way:
let $\sigma = e_1 \dotsm e_n$ be a sequence of events in a PPES $ \delta' = \left( \delta, \lessdot \right) $. We call $ \sigma $ a trace of $ \delta' $ iff it is
\begin{inparaenum}[1.)]
	\item a trace of $ \delta $, and
	\item satisfies the following constraint:
\end{inparaenum}
\vspace{-1em}
\begin{equation}
	\label{eq:PPESTraceDef}
	\forall i<n \logdot \forall e_{j},e_{h} \in \bar{\sigma} \logdot e_{j} \neq e_{h} \wedge e_{j}, e_{h} \in \en{\delta}{\sigma_{i}} \wedge e_{h} \lessdot e_{j} \implies j<h
	\vspace{-0.2em}
\end{equation}

For example, the sequence $ebad$ is a trace of Figure~\ref{fig:PESExamples}~(a) but not of Figure~\ref{fig:PESExamples}~(b) due to priority. Let us denote the set of traces of a structure as $ \traces{\delta} $.
By definition, the traces of a PPES $ \delta' = \left( \delta, \lessdot \right) $ are a subset of the traces of the Prime ES $ \delta $.

\begin{prop}
	\label{prop:PPESSubTraces}
	$ \traces{\delta, \lessdot} \subseteq \traces{\delta} $.
\end{prop}

If we analyze Figure~\ref{fig:PESExamples}~(a) and (b) we observe that, because of the conflict relation, no trace can contain both $c$ and $d$. And even without the priority relation the enabling relation ensures that $e$ always has to precede $d$. Since neither $c$ and $d$ nor $e$ and $d$ can be enabled together, \ie do never compete, applying the priority relation between them is useless or trivial. Indeed we can always reduce the priority relation by dropping all pairs between events that are under $\leq$ or $\#$ without affecting the set of traces.

\begin{thm}
	\label{thm:PPESEquivalence}
	Let $\left( E, \#, \leq, \labeling{}, \lessdot \right)$ be an PPES, and let 
$\lessdot' \deff \lessdot \setminus \Set{ (e, e') \mid e' \# e \vee e' \leq e \vee e \leq e' }$. Then:
	\vspace{-0.3em}
	\begin{align*}
		\traces{E, \#, \leq, \labeling{}, \lessdot} = \traces{E, \#, \leq, \labeling{}, \lessdot'}
	\end{align*}
\end{thm}

\begin{proof}
	Straightforward from the Definitions of traces, \eqref{eq:PPESTraceDef}, and \eqref{eq:PrimeEnablingDef}.
\end{proof}

\noindent
Figure~\ref{fig:PESExamples}~(c) shows the result of dropping the priority pairs that are redundant in Figure~\ref{fig:PESExamples}~(b). Note that after dropping all redundant pairs, there is no overlapping, neither between the priority and the enabling relation, nor between the priority and the conflict relation. The following theorem insures minimality of reduction.

\begin{thm}
	\label{thm:minimalityPPES}
	Let $\left( E, \#, \leq, \labeling{}, \lessdot \right)$ be an PPES, let 
$\lessdot' \deff \lessdot \setminus \Set{ (e, e') \mid e' \# e \vee e' \leq e \vee e \leq e' }$, and $ \lessdot'' \subset \lessdot' $. Then $ \traces{E, \#, \leq, \labeling{}, \lessdot} \neq \traces{E, \#, \leq, \labeling{}, \lessdot''} $.
\end{thm}

\begin{proof}
	Straightforward from the Definitions of traces, \eqref{eq:PPESTraceDef}, and \eqref{eq:PrimeEnablingDef}.
\end{proof}

\noindent
This result is good for a modeler, since it implies unambiguity about whether a priority relation affects the behavior or not. In other words, after dropping all the redundant priority pairs, the remaining priority pairs always lead to pre-emption, limit concurrency and narrow down the possible traces.
This is not the case for the following ESs, since they offer other causality models.


\section{Priority in Bundle Event Structures}
\label{sec:PBES}

Prime ESs are simple but also limited. Conflict heredity and the enabling relation of Prime ESs do not allow to describe some kind of optional or conditional enabling of events. Bundle ESs|among others|were designed to overcome these limitations \cite{Langerak:Thesis}. Here an event can have different causes, \ie they allow for disjunctive causality.

In Bundle ESs the conflict relation is as in Prime ESs an irreflexive and symmetric relation, but the enabling relation offers some optionality, based on bundles.
A \emph{bundle} $ \left( X, e \right) $, also denoted by $ X \mapsto e $, consists of a bundle set $ X $ and the event $ e $ it enables. A \emph{bundle set} is a set of events that are pairwise in conflict. There can be several bundles $ \left( X_1, e \right), \ldots, \left( X_n, e \right) $ for the same event $ e $. So|instead of a set of events as in Prime ESs|an event $e$ in Bundle ESs is enabled by a set $ \Set{ X_1, \ldots, X_n } $ of bundle sets.

When one event of a set $ X_i $ takes place, then the bundle $ X_i \mapsto e $ is said to be satisfied; and for $ e $ to be enabled all its bundles must be satisfied. In Bundle ESs (and also Extended Bundle ESs) no more than one event out of each set $ X_i $ can take place; this leads to causal unambiguity \cite{Langerak97causalambiguity}. But since a bundle set can be satisfied by any of its members, this yields disjunctive causality and gives flexibility in enabling.

\begin{mydef}
	\label{def:BES}
	A \emph{Bundle Event Structure (BES)} is a quadruple $ \beta = \left( E, \#, \mapsto, \labeling{} \right) $ where:
	\begin{compactitem}
		\item $ E $, a set of \emph{events}
		\item $ \# \subseteq E \times E $, an irreflexive symmetric relation (the \emph{conflict} relation)
		\item $ \mapsto \, \subseteq \Powerset{E} \times E $, the \emph{enabling} relation
		\item $ \labeling{} : E \to Act $, a \emph{labeling} function
	\end{compactitem}
	that additionally satisfies the following constraint:
	\vspace{-0.8em}
	\begin{equation*}
		\vspace{-0.5em}
		\label{eq:StabilityConstraintBES}
		\text{\textbf{Stability:}} \quad \quad \forall X \subseteq E \logdot \forall e \in E \logdot X \mapsto e \implies  \left( \forall e_{1}, e_{2} \in X \logdot e_{1} \neq e_{2} \implies e_{1} \# e_{2} \right) \tag{SC}
	\end{equation*}
\end{mydef}

Figure~\ref{fig:PBESExample}~(a) shows an example of a BES. The solid arrows denote causality, \ie reflect the enabling relation, where the bar between the arrows shows that they belong to the same bundle and the dashed line denotes again a mutual conflict. Thus there are two bundles in this example, namely the singleton $ \Set{ a } \mapsto b $ and $ \Set{ b, c } \mapsto d $. As required by \eqref{eq:StabilityConstraintBES} we have $ b \# c $ and $ c \# b $.

\begin{figure}[t]
	\centering
	\vspace{-1em}
	\begin{tikzpicture}[bend angle=45]
		\node (a1)	at (0, 2)		{};
		\node (b1)	at (0, 1)		{};
		\node (c1)	at (1, 1)		{};
		\node (d1)	at (0.5, 0)		{};
		\node (a)	at (0.5, -0.8)	{(a)};
		\fill (a1) circle (2.5pt) node [right]	{$a$};
		\fill (b1) circle (2.5pt) node [left]	{$b$};
		\fill (c1) circle (2.5pt) node [right]	{$c$};
		\fill (d1) circle (2.5pt) node [below]	{$d$};
		\draw[enabling]		(a1) -- (b1);
		\draw\enabling{2}{2}	(b1) -- (d1);
		\draw\enabling{2}{2}	(c1) -- (d1);
		\draw[conflict]		(b1) -- (c1);
		\draw[thick] 		(0.25, 0.5) -- (0.75, 0.5);
		\node (a2)	at (3, 2)		{};
		\node (b2)	at (3, 1)		{};
		\node (c2)	at (4, 1)		{};
		\node (d2)	at (3.5, 0)		{};
		\node (b)	at (3.5, -0.8)	{(b)};
		\fill (a2) circle (2.5pt) node [right]	{$a$};
		\fill (b2) circle (2.5pt) node [left]	{$b$};
		\fill (c2) circle (2.5pt) node [right]	{$c$};
		\fill (d2) circle (2.5pt) node [below]	{$d$};
		\draw[enabling]		(a2) -- (b2);
		\draw\enabling{2}{2}	(b2) -- (d2);
		\draw\enabling{2}{2}	(c2) -- (d2);
		\draw[conflict]		(b2) -- (c2);
		\draw[thick] 		(3.25, 0.5) -- (3.75, 0.5);
		\draw[priority] (b2) .. controls (3.3, 1.7) and (3.7, 1.7) .. (c2);
		\draw[priority] (c2) to [bend left]	(d2);
		\draw[priority] (d2) .. controls (1.8, 0.6) and (2.6, 2) .. (a2);
		\node (a3)	at (6, 2)		{};
		\node (b3)	at (6, 1)		{};
		\node (c3)	at (7, 1)		{};
		\node (d3)	at (6.5, 0)		{};
		\node (c)	at (6.5, -0.8)	{(c)};
		\fill (a3) circle (2.5pt) node [right]	{$a$};
		\fill (b3) circle (2.5pt) node [left]	{$b$};
		\fill (c3) circle (2.5pt) node [right]	{$c$};
		\fill (d3) circle (2.5pt) node [below]	{$d$};
		\draw[enabling]		(a3) -- (b3);
		\draw\enabling{2}{2}	(b3) -- (d3);
		\draw\enabling{2}{2}	(c3) -- (d3);
		\draw[conflict]		(b3) -- (c3);
		\draw[thick] 		(6.25, 0.5) -- (6.75, 0.5);
		\draw[priority] (d3) .. controls (4.8, 0.6) and (5.6, 2) .. (a3);
	\end{tikzpicture}
	\vspace{-1em}
	\caption{A Bundle ES without priority in (a), with priority in (b), and after dropping redundant priority pairs in (c).}
	\label{fig:PBESExample}
\end{figure}
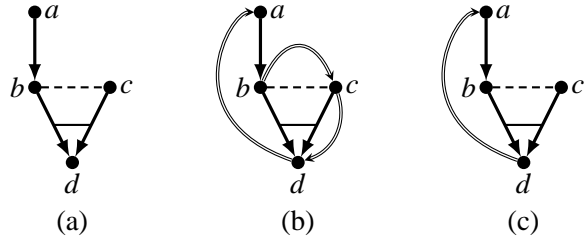

As in Prime ES, let $ \sigma = e_1 \dotsm e_n $ be a sequence of events and $ \bar{\sigma} = \Set{ e_1, \ldots, e_n } $ such that $ \bar{\sigma} \subseteq E $. We use $ \en{\beta}{\sigma} $ to refer to the set of events enabled by $ \sigma $:
\begin{equation}
	\label{eq:BESEnablingDef}
	\en{\beta}{\sigma} \deff \Set{ e \in (E \setminus \bar{\sigma} ) \mid \left( \forall X \subseteq E \logdot X \mapsto e \implies X \cap \bar{\sigma} \neq \emptyset \right) \wedge \left( \nexists e' \in \bar{\sigma} \logdot e \# e' \right) }
\end{equation}
Then the sequence $\sigma = e_{1} \dotsm e_{n}$ is called a \emph{trace} of $ \beta $ iff
\begin{equation}
	\label{eq:BESTraceDef}
	\forall i \leq n \logdot e_{i} \in \en{\beta}{\sigma_{i-1}}
\end{equation}
We denote the set of all valid traces in $ \beta $ as $ \traces{\beta} $.
A set of events \(C \subseteq E\) is called a \emph{configuration} of an BES $ \beta $ iff $ \exists \sigma \in \traces{\beta} \logdot C = \bar{\sigma} $. Let $ \configurations{\beta} $ denote the set of configurations of $ \beta $.

Definition~\eqref{eq:BESEnablingDef} ensures that for an event $ e $ to be enabled, one event out of each pointing bundle set $ X $ with $ X \mapsto e $ is necessary. In a trace, if there is one event of a bundle set, then we denote the corresponding bundle as \emph{satisfied}. Remember that, because of the stability condition, no more than one event out of each bundle set can take place.
In Figure~\ref{fig:PBESExample}, the sequence $ bd $ is not a trace since $ \Set{ a } \mapsto b $ has never been satisfied. On the other hand, $ abcd $ is not a trace either, since $ b $ conflicts with $ c $. While $ a $, $ c $, $ ab $, $ ac $, $ ca $, $ cd $, $ abd $, $ acd $, $ cad $, and $ cda $ are all traces.

The following lemma proves that whenever an event $ e $ is in a bundle set $ X $ pointing to $ e' $, \ie such that $ X \mapsto e' $, then $ e $ and $ e' $ cannot be enabled together.

\begin{lma}
	\label{lma:BundleNecessaryCause}
	Let $ \beta = \left( E, \#, \mapsto, \labeling{} \right) $ be an BES, and let $ e, e' \in E $ such that $ \exists X \subseteq E \logdot e \in X \land X \mapsto e' $. Then:
	\vspace{-0.3em}
	\begin{align*}
		\forall \sigma = e_{1} \dotsm e_{n} \in \traces{\beta} \logdot \nexists i < n \logdot e,e' \in \en{\beta}{\sigma_i}
	\end{align*}
\end{lma}

\begin{proof}
	Let $ \sigma = e_{1} \dotsm e_{n} \in \traces{\beta} $ and $ X \subseteq E $ such that $ e \in X \land X \mapsto e' $.
	Assume $ e \in \en{\beta}{\sigma_i} $ for some $ i < n $. Then:
	\vspace{-1em}
	\begin{align*}
		& e \in \en{\beta}{\sigma_i} \wedge e \in X \wedge X \mapsto e' \stackrel{\eqref{eq:StabilityConstraintBES}}{\implies} e \in \en{\beta}{\sigma_i} \wedge e \in X \wedge X \mapsto e' \wedge \left( \forall e'' \in \left( X \setminus \Set{ e } \right) \logdot e \# e'' \right)\\
		& \stackrel{\eqref{eq:BESEnablingDef}}{\implies} X \mapsto e' \wedge \left( \forall e'' \in X \logdot e'' \notin \bar{\sigma_i} \right) \stackrel{\eqref{eq:BESEnablingDef}}{\implies} e' \notin \en{\beta}{\sigma_i}
	\end{align*}
	\vspace{-2em}\\
	Hence $ \nexists i < n \logdot e,e' \in \en{\beta}{\sigma_i} $.
\end{proof}

\subsection{Labeled Partially Ordered Sets}
\label{sec:lposetBES}

Labeled partially ordered sets, abbreviated as lposets, are used as a semantical model for different kinds of ESs and other concurrency models (see \eg \cite{posetsForConfigurations}). In contrast to configurations, lposets do not only record the set of events that happened so far, but also reflect the precedence relations between these events. Here, we use them to describe the semantics of BESs (as well as of EBESs and DESs in the next sections). An lposet consists of a set, a partial order over this set, and a labeling function.

\begin{mydef}
	A \emph{labeled partially ordered set (lposet)} is a triple \( \left\langle A, \leq, \operatorname{f} \right\rangle\) where:
	\begin{compactitem}
		\item \(A\), a finite set of \emph{events}
		\item \(\leq\), a \emph{partial order} over \(A\)
		\item \(\operatorname{f} : A \to Act\), a \emph{labeling} function
	\end{compactitem}
\end{mydef}

We use $ \eta $ to denote the empty lposet $ \left\langle \emptyset, \emptyset, \emptyset \right\rangle $. A non-empty lposet $ \left\langle A, \leq, \operatorname{f} \right\rangle $ is visualized by a box containing all the events of $ A $ and where two events $ e_1 $ and $ e_2 $ are related by an arrow iff $ e_1 \leq e_2 $, where reflexive and transitive arrows are usually omitted. Figure~\ref{fig:lposetFamily} depicts several lposets, where \eg the top right box visualizes the lposet $ \left\langle \Set{ a, b, d }, \Set{ \left( a, b \right), \left( b, d \right) }, \operatorname{f} \right\rangle $ for some (not visualized) labeling function~$\operatorname{f}$.

An lposet describes the semantics of a BES for a specific set of events. To describe the semantics of the entire BES, families of lposets are used. These families consist of several lposets that are related by a prefix relation on lposets \cite{Katoen:Thesis}:
\begin{equation*}
	\left\langle A, \leq, \operatorname{f} \right\rangle \text{ is a \emph{prefix} of } \left\langle A', \leq', \operatorname{f}' \right\rangle \iff A \subseteq A' \land \leq \; = \left( \leq' \cap \left( A' \times A \right) \right) \land f = \reducedto{f'}{A}
\end{equation*}
Now a family $\mathcal{P}$ of lposets is defined as a non-empty set of lposets that is downward closed under the lposet-prefix relation. It is shown in \cite{Katoen:Thesis} that a family of lposets along with the prefix relation is itself a partially ordered set.
As investigated by Rensink in \cite{posetsForConfigurations}, families of lposets (even posets) form a convenient underlying model for models of concurrency like BESs (or EBESs).

In order to define the lposets of a BES $ \beta $, we build a partially ordered set (poset) over a configuration $ C \in \configurations{\beta} $. We define the partial order as a precedence relation $ \prec_{C} \; \subseteq C \times C $ between events as follows:
\begin{equation}
	\label{eq:lposetBES}
	e \prec_{C} e' \iff \exists X \subseteq E \logdot e \in X \land X \mapsto e'
\end{equation}
and define $ \preceq_{C} $ as the reflexive and transitive closure of $ \prec_{C} $. It is proved in \cite{Langerak:Thesis} that $ \preceq_{C} $ is a partial order over $ C $. Finally, by adding the labeling function $ \reducedto{l}{C} $, the triple $ \left\langle C, \preceq_{C}, \reducedto{l}{C} \right\rangle $ is an lposet.
We call $ \lposets{\beta} $ the set of all lposets defined on $ \configurations{\beta} $.
Figure~\ref{fig:lposetFamily} shows the largest family of lposets for the example in Figure~\ref{fig:PBESExample}~(a), where the arrows between lposets denote the prefix relation.

\begin{figure}[t]
	\centering
	\vspace{-1em}
	\begin{tikzpicture}
		\node			(l1)		at (0, 1.25)		{$ \eta $};
		\node[lposet]	(l2)		at (1.5, 2)		{$ a $};
		\node[lposet]	(l3)		at (1.5, 0.5)	{$ c $};
		\node[lposet]	(l4)		at (3.5, 2)		{
			\begin{tikzpicture}
				\node (a) at (0, 0)		{$ a $};
				\node (b) at (0.8, 0.04)	{$ b $};
				\draw[->, thick] (a) -- (0.6, 0);
			\end{tikzpicture}
		};
		\node[lposet]	(l5)		at (3.5, 1)		{
			\begin{tikzpicture}
				\node (a) at (0, 0.4)	{$ a $};
				\node (c) at (0, 0)		{$ c $};
			\end{tikzpicture}
		};
		\node[lposet]	(l6)		at (3.5, 0)		{
			\begin{tikzpicture}
				\node (c) at (0, 0)		{$ c $};
				\node (d) at (0.8, 0.04)	{$ d $};
				\draw[->, thick] (c) -- (0.6, 0);
			\end{tikzpicture}
		};
		\node[lposet]	(l7)		at (6.5, 2)	{
			\begin{tikzpicture}
				\node (a) at (0, 0)		{$ a $};
				\node (b) at (0.8, 0.04)	{$ b $};
				\node (d) at (1.6, 0.04)	{$ d $};
				\draw[->, thick] (a) -- (0.6, 0);
				\draw[->, thick] (1, 0) -- (1.4, 0);
			\end{tikzpicture}
		};
		\node[lposet]	(l8)		at (6.5, 0.5)	{
			\begin{tikzpicture}
				\node (a) at (0.4, 0.4)	{$ a $};
				\node (c) at (0, 0)		{$ c $};
				\node (d) at (0.8, 0.04)	{$ d $};
				\draw[->, thick] (c) -- (0.6, 0);
			\end{tikzpicture}
		};
		\draw[-stealth, thick] (l1) -- (l2);
		\draw[-stealth, thick] (l1) -- (l3);
		\draw[-stealth, thick] (l2) -- (l4);
		\draw[-stealth, thick] (l2) -- (l5);
		\draw[-stealth, thick] (l3) -- (l5);
		\draw[-stealth, thick] (l3) -- (l6);
		\draw[-stealth, thick] (l4) -- (l7);
		\draw[-stealth, thick] (l5) -- (l8);
		\draw[-stealth, thick] (l6) -- (l8);
	\end{tikzpicture}
	\vspace{-0.5em}
	\caption{The family of lposets of the BES in Figure~\ref{fig:PBESExample}~(a).}
	\label{fig:lposetFamily}
\end{figure}
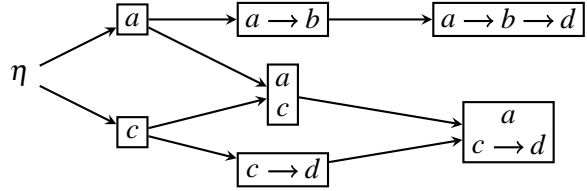

As proved in \cite{Langerak:Thesis}, each linearization (obeying the defined precedence relations) of a given lposet built from an BES (or EBES) yields an event trace of that structure.

\subsection{Adding Priority to BESs}
\label{sec:PriorityBES}

Again we add priority $ \lessdot \subseteq E \times E $ to EBESs as a binary acyclic relation between events such that, whenever two events are enabled together, the one with the higher priority pre-empts the other. We denote $ \beta' = \left( \beta, \lessdot \right) = \left( E, \#, \mapsto, \labeling{}, \lessdot \right) $ as \emph{prioritized Bundle ES (PBES)}.
Figure~\ref{fig:PBESExample}~(b) illustrates a prioritized version of the BES in Figure~\ref{fig:PBESExample}~(a).

Also the semantics of $ \lessdot $ is defined similarly to Prime ESs. A sequence of events $ \sigma = e_1 \dotsm e_n $ is a trace of $ \left( \beta, \lessdot \right) $ iff
\begin{inparaenum}[1.)]
	\item $ \sigma \in \traces{\beta} $ and
	\item $ \sigma $ satisfies the following constraint:
\end{inparaenum}
\begin{equation}
	\label{eq:PBESTraceDef}
	\forall i < n \logdot \forall e_j, e_h \in \bar{\sigma} \logdot e_j \neq e_h \wedge e_j, e_h \in \en{\beta}{\sigma_i} \wedge e_h \lessdot e_j \implies j < h
\end{equation}
Again the traces of a PBES $ \left( \beta, \lessdot \right) $ are a subset of the traces of the corresponding BES $ \beta $.
\begin{prop}
	$ \traces{ \beta, \lessdot} \subseteq \traces{\beta} $.
\end{prop}

For example the sequence $ cad $ is a trace of the BES in Figure~\ref{fig:PBESExample}~(a), but it is not a trace of the PBES in Figure~\ref{fig:PBESExample}~(b).
Of course a larger priority relation filters more traces out than a smaller one.

\begin{lma}
	\label{lma:traceFilteringPBES}
	Let $ \left( \beta, \lessdot \right) $ and $ \left( \beta, \lessdot' \right) $ be two PBES with $ \lessdot' \subseteq \lessdot $. Then $ \traces{\beta, \lessdot} \subseteq \traces{\beta, \lessdot'} $.
\end{lma}

\begin{proof}
	Straightforward from the Definition of traces, \eqref{eq:PBESTraceDef}, and $ \lessdot' \subseteq \lessdot $.
%
%
\end{proof}

We adapt the notion of configuration to prioritized BESs such that $ \sigma \in \configurations{\beta, \lessdot} $ for a PBES $ \left( \beta, \lessdot \right) $ iff $ \exists \sigma \in \traces{\beta, \lessdot} \logdot C = \bar{\sigma} $.
In Section~\ref{sec:lposetBES} we define the semantics of BESs by families of lposets. Unfortunately doing the same for PBESs is not that simple. Consider the lposet
\begin{tikzpicture}
	\node[lposet] (l1) at (0, 0) {
		\begin{tikzpicture}
			\node (a) at (0, 0)		{$ a $};
			\node (c) at (0.4, 0)	{$ c $};
			\node (d) at (1.2, 0.04)	{$ d $};
			\draw[->, thick] (c) -- (1, 0);
		\end{tikzpicture}
	};
\end{tikzpicture}
of the BES $ \beta $ in Figure~\ref{fig:PBESExample}~(a). According to Figure~\ref{fig:PBESExample}~(b), $ d $ has a higher priority than $ a $, \ie $ a \lessdot d $. Hence
\begin{tikzpicture}
	\node[lposet] (l1) at (0, 0) {
		\begin{tikzpicture}
			\node (a) at (0, 0)		{$ a $};
			\node (c) at (0.4, 0)	{$ c $};
			\node (d) at (1.2, 0.04)	{$ d $};
			\draw[->, thick] (c) -- (1, 0);
		\end{tikzpicture}
	};
\end{tikzpicture}
does not describe the semantics of the PBES $ \left( \beta, \lessdot \right) $ with respect to the Configuration $ \Set{ a, c, d } $, because $ cad \in \traces{\beta} $ but $ cad \notin \traces{\beta, \lessdot} $. In fact we cannot describe the semantics of PBESs by a family of lposets as depicted in Figure~\ref{fig:lposetFamily}. Instead, to describe the semantics of $ \left( \beta, \lessdot \right) $ with respect to $ \Set{ a, c, d } $ we need the two different lposets
\begin{tikzpicture}
	\node[lposet] (l1) at (0, 0) {
		\begin{tikzpicture}
			\node (a) at (0, 0)		{$ a $};
			\node (c) at (0.8, 0)	{$ c $};
			\node (d) at (1.6, 0.04)	{$ d $};
			\draw[->, thick] (a) -- (c);
			\draw[->, thick] (c) -- (1.4, 0);
		\end{tikzpicture}
	};
\end{tikzpicture}
and
\begin{tikzpicture}
	\node[lposet] (l1) at (0, 0) {
		\begin{tikzpicture}
			\node (c) at (0, 0)		{$ c $};
			\node (d) at (0.8, 0.04)	{$ d $};
			\node (a) at (1.6, 0)	{$ a $};
			\draw[->, thick] (c) -- (0.6, 0);
			\draw[->, thick] (1, 0) -- (a);
		\end{tikzpicture}
	};
\end{tikzpicture}.

The enabling relation defines precedence between events as used for $ \prec_C $ in \eqref{eq:lposetBES}, whereas priority rather defines some kind of conditional precedence. Priority affects the semantics only if the related events are enabled together. Thus the same problem with the definition of lposets appears for all kinds of Event Structures that are extended by priority. We leave the problem on how to fix the definition of lposets as future work.

\subsection{Priority versus Enabling and Conflict}

Again, as in Section \ref{sec:PPES}, we can reduce the priority relation by removing redundant pairs, \ie pairs that due to the enabling or conflict relation do not affect the semantics of the PBES. First we can|as already done in PPES|remove a priority pair $ e \lessdot e' $ or $ e' \lessdot e $ between an event $ e $ and its cause $ e' $, because an event and its cause are never enabled together. Therefore \eg the pair $ d \lessdot c $ in Figure~\ref{fig:PBESExample}~(b) is redundant because of $ \Set{ b, c } \mapsto d $. Also a priority pair $ e \lessdot e' $ between two events that are in conflict is redundant, because these conflicting events never occur in the same trace. Consider for example the events $ b $ and $ c $ in Figure~\ref{fig:PBESExample}~(b). Because of $ b \# c $ the pair $ c \lessdot b $ is redundant.

\begin{thm}
	\label{thm:PriorityReduction}
	Let $ \left( \beta, \lessdot \right) = \left( E, \leadsto, \mapsto, \labeling{}, \lessdot \right) $ be a PBES and
	\begin{align*}
		\lessdot' = \lessdot \setminus \Set{ \left( e, e' \right), \left( e', e \right) \mid e \# e' \vee \left( \exists X \subseteq E \logdot e \in X \wedge X \mapsto e' \right) }.
	\end{align*}
	Then $ \traces{\beta, \lessdot} = \traces{\beta, \lessdot'} $.
\end{thm}

\begin{proof}
	$ \traces{\beta, \lessdot} \subseteq \traces{\beta, \lessdot'} $ follows from Lemma~\ref{lma:traceFilteringPBES}.
	
	To show $ \traces{\beta, \lessdot'} \subseteq \traces{\beta, \lessdot} $, assume a trace $ \sigma = e_1 \dotsm e_n \in \traces{\beta, \lessdot'} $. We have to show that $ \sigma \in \traces{\beta, \lessdot} $, \ie that $ \sigma \in \traces{\beta} $ and that $ \sigma $ satisfies Condition~\eqref{eq:PBESTraceDef}. $ \sigma \in \traces{\beta} $ follows from $ \sigma \in \traces{\beta, \lessdot'} $ by the Definition of traces. $ \sigma $ satisfies Condition~\eqref{eq:PBESTraceDef} when $ \forall i < n \logdot \forall e_j, e_h \in \bar{\sigma} \logdot e_j \neq e_h \wedge e_j, e_h \in \en{\beta}{\sigma_i} \wedge e_h \lessdot e_j \implies j < h $. Let us fix $ i < n $ and $ e_j, e_h \in \bar{\sigma} $. Assume $ e_j \neq e_h $, $ e_j, e_h \in \en{\beta}{\sigma_i} $, and $ e_h \lessdot e_j $. It remains to prove that $ j < h $. Because of the Definition of $ \lessdot' $, there are three cases for $ e_h \lessdot e_j $:
	\vspace{-0.3em}
	\begin{description}
		\item[Case $ e_h \lessdot' e_j $:] Then $ e_h \lessdot' e_j \wedge e_j, e_h \in \bar{\sigma} \wedge \sigma \in \traces{\beta, \lessdot'} \stackrel{\eqref{eq:PBESTraceDef}}{\implies} j < h $.
		\item[Case $ e_j \# e_h \vee e_h \# e_j $:] This case is not possible, because it is in contradiction to \eqref{eq:BESTraceDef}, \eqref{eq:BESEnablingDef}, and $ e_j, e_h \in \bar{\sigma} $.
		\item[Case $ \exists X \subseteq E \logdot \left( e_h \in X \wedge X \mapsto e_j \right) \lor \left( e_j \in X \wedge X \mapsto e_h \right) $:] This case is not possible, because it is in contradiction to $ e_j, e_h \in \en{\beta}{\sigma_i} $ and Lemma~\ref{lma:BundleNecessaryCause}.
	\end{description}
	\vspace{-1.7em}
\end{proof}

\noindent
Note that priority is redundant for all pairs of events that are directly related by the bundle enabling relation or the conflict relation regardless of the direction of the priority pair. We say that this reduction is done at the structure level, since it is done \wrt the relations which are part of the Event Structure.

In PPESs enabling is a transitive relation and we can drop all priority pairs between events that are related by enabling. In the case of PBESs neither conflict nor enabling are transitive relations. For example in the event structure
	\begin{tikzpicture}[bend angle=45]
		\node (e1)	at (0, 0)	{};
		\node (e2)	at (1, 0)	{};
		\node (e3)	at (2, 0)	{};
		\fill (e1) circle (2.5pt) node [above]	{$ e_1 $};
		\fill (e2) circle (2.5pt) node [above]	{$ e_2 $};
		\fill (e3) circle (2.5pt) node [above]	{$ e_3 $};
		\draw[conflict]		(e1) -- (e2);
		\draw[conflict]		(e2) -- (e3);
	\end{tikzpicture}
(which can be both; a PES as well as a BES) we have $ e_1 \# e_2 $ and $ e_2 \# e_3 $ but not $ e_1 \# e_3 $. Accordingly we cannot drop a priority pair $ e_1 \lessdot e_3 $ because else the sequence $ e_1e_3 $ becomes a trace.

However in PPESs enabling is transitive, so whenever $ e_1 \leq e_2 $ and $ e_2 \leq e_3 $ there is $ e_1 \leq e_3 $ and we can also drop priority pairs relating $ e_1 $ and $ e_3 $ (compare \eg with $ e $, $ a $, and $ d $ in Figure~\ref{fig:PESExamples}). In PBES the situation is different. For the PBES in Figure~\ref{fig:PBESExample} we have $ \Set{ a } \mapsto b $ and $ \Set{ b, c } \mapsto d $ but $ d $ does not necessarily depend on $ a $ and thus we cannot drop the pair $ a \lessdot d $ since $ cad \notin \traces{\beta, \lessdot} $.
Unfortunately this means that we do not necessarily drop the whole redundancy in priority if we reduce the priority relation as described in Theorem~\ref{thm:PriorityReduction}. For example $ e_1 \lessdot e_3 $ is redundant in $ \left( \Set{ e_1, e_2, e_3 }, \emptyset, \Set{ \Set{ e_1 } \mapsto e_2, \Set{ e_2 } \mapsto e_3 }, \labeling{}, \Set{ e_1 \lessdot e_3 } \right) $, because in this special case $ e_1 $ is indeed a necessary cause for $ e_3 $.
Thus for PBESs $ \lessdot' $ is not necessarily minimal, \ie we cannot prove $ \forall \lessdot'' \subset \lessdot' \logdot \traces{\left( E, \leadsto, \mapsto, \labeling{}, \lessdot \right)} \neq \traces{\left( E, \leadsto, \mapsto, \labeling{}, \lessdot'' \right)} $ as we have done in Theorem \ref{thm:minimalityPPES} for PPESs.

For the PBES in Figure~\ref{fig:PBESExample} the reduction described in Theorem~\ref{thm:PriorityReduction} indeed suffices to remove all redundant priority pairs. The result is presented in Figure~\ref{fig:PBESExample}~(c).

\subsection{Priority versus Precedence}

In order to identify some more redundant priority pairs we consider configurations and lposets. If we analyze for example the configurations $ \Set{ a, b, c } $ and $ \Set{ a, c, d } $ of the PBES in Figure~\ref{fig:PBESExample}, we observe that, because of $ \Set{ a } \mapsto b $ and $ \Set{ b, c } \mapsto d $, the priority pair $ a \lessdot d $ is redundant in the first configuration while it is not in the second one. Thus, in some cases, \ie with respect to some configurations (or lposets), we can also ignore priority pairs of events that are indirectly related by enabling. Since such a redundancy is relative to specific configurations and their traces, and since dropping priority pairs affects the whole set of traces obtained from a ES, we use the term ``ignorance'' rather than ``dropping'' for distinction, and we say that this ignorance is done at the configuration level. Priority ignorance is necessary while linearizing configurations and trying to obtain traces. 

The cases in which priority pairs are redundant with respect to some configuration $ C $ are already well described by the precedence relation $ \preceq_C $, \ie we can identify redundant priority pairs easily from the lposets for $ C $.
Note that in BESs (and also EBESs) each configuration leads to exactly one lposet.
The priority pair $ a \lessdot d $ is obviously redundant in the case of
\begin{tikzpicture}
	\node[lposet] (l1) at (0, 0) {
		\begin{tikzpicture}
			\node (a) at (0, 0)		{$ a $};
			\node (b) at (0.8, 0.04)	{$ b $};
			\node (d) at (1.6, 0.04)	{$ d $};
			\draw[->, thick] (a) -- (0.6, 0);
			\draw[->, thick] (1, 0) -- (1.4, 0);
		\end{tikzpicture}
	};
\end{tikzpicture}
but not in the case of
\begin{tikzpicture}
	\node[lposet] (l1) at (0, 0) {
		\begin{tikzpicture}
			\node (a) at (0, 0)		{$ a $};
			\node (c) at (0.4, 0)	{$ c $};
			\node (d) at (1.2, 0.04)	{$ d $};
			\draw[->, thick] (c) -- (1, 0);
		\end{tikzpicture}
	};
\end{tikzpicture}.

To formalize this let $ \reducedto{\traces{\beta, \lessdot}}{C} \deff \Set{ \sigma \mid \sigma \in \traces{\beta, \lessdot} \wedge \bar{\sigma} = C } $ be the set of traces over the configuration $ C \subseteq E $ for some BES $ \beta = \left( E, \leadsto, \mapsto, \labeling{} \right) $. Thus $ \reducedto{\traces{\beta, \lessdot}}{C} $ consists of all the traces of $ \traces{\beta, \lessdot} $ that are permutations of the events in $ C $. Then for all configurations $ C $ all priority pairs $ e \lessdot e' $ such that $ e' \preceq_C e $ or $ e \preceq_C e' $ can be ignored.

\begin{thm}
	\label{thm:PriorityIgnoranceBES}
	Let $ \left( \beta, \lessdot \right) $ be a PBES, $ \left\langle C, \preceq_C, \labeling{} \right\rangle \in \lposets{\beta} $, and $ \lessdot' \deff \lessdot \setminus \Set{ \left( e, e' \right) \mid e' \preceq_C e \vee e \preceq_C e' } $. Then:
	\begin{align*}
		\reducedto{\traces{\beta, \lessdot}}{C} \; = \reducedto{\traces{\beta, \lessdot'}}{C}
	\end{align*}
\end{thm}

\begin{proof}
	Note that by induction on $ \preceq $ and Lemma~\ref{lma:BundleNecessaryCause}, $ e_j \preceq_C e_h $ as well as $ e_h \preceq_C e_j $ imply that $ e_j $ and $ e_h $ cannot be enabled together in a trace of $ \reducedto{\traces{\beta}}{C} $. With this argument the proof is straightforward from the definitions of traces, $ \lessdot' $, traces over a configuration, Lemma~\ref{lma:traceFilteringPBES}, and \eqref{eq:PBESTraceDef}.
\end{proof}

\noindent
Consider once more the PBES $ \left( \beta, \lessdot \right) $ of Figure~\ref{fig:PBESExample} with respect to the configuration $ \Set{ a, b, d } $. We have $ \Set{ a } \mapsto b $, $ \Set{ b, c } \mapsto d $, and $ a \lessdot d $. As explained before we cannot drop the priority pair $ a \lessdot d $, because of the sequence $ cad \notin \traces{\varepsilon, \lessdot} $. However with Theorem~\ref{thm:PriorityIgnoranceBES} we can ignore $ a \lessdot d $ for the semantics of $ \left( \beta, \lessdot \right) $ if we limit our attention to $ \Set{ a, b, d } $, because $ \reducedto{\traces{\beta, \lessdot}}{\Set{ a, b, d }} = \Set{ abd } = \reducedto{\traces{\beta}}{\Set{ a, b, d }} $.

For PBESs ignorance ensures that $ \lessdot' $ is minimal with respect a configuration $ C $.

\begin{thm}
	\label{thm:minimalityPBES}
	Let $ \left( \beta, \lessdot \right) $ be a PBES, $ \left\langle C, \preceq_C, \labeling{} \right\rangle \in \lposets{\beta} $ for some configuration $ C \in \configurations{\beta, \lessdot} $, $ \lessdot' \deff \lessdot \setminus \Set{ \left( e, e' \right) \mid e' \preceq_C e \vee e \preceq_C e' } $, and $ \lessdot'' \subset \lessdot' $. Then $ \reducedto{\traces{\beta, \lessdot}}{C} \; \neq \reducedto{\traces{\beta, \lessdot''}}{C} $.
\end{thm}

\begin{proof}
	Because of $ \lessdot'' \subset \lessdot' $, there are some $ e, e' \in E $ such that $ e \lessdot e' $ but $ e \not\!\!\!\lessdot' e' $, $ e' \not\preceq_C e $, and $ e \not\preceq_C e' $.
	Note that each linearization of a given lposet that respects the precedence relation is a trace \cite{Langerak:Thesis}. Thus $ e' \not\preceq_C e $ and $ e \not\preceq_C e' $ imply that $ \reducedto{\traces{\beta, \lessdot''}}{C} $ contains a trace such that $ e $ and $ e' $ are enabled together and $ e $ precedes $ e' $. Because of $ e \lessdot e' $ such a trace cannot be contained in $ \reducedto{\traces{\beta, \lessdot}}{C} $. So $ \reducedto{\traces{\beta, \lessdot}}{C} \; \neq \reducedto{\traces{\beta, \lessdot''}}{C} $.
\end{proof}

\noindent
In the following two sections we consider two extensions of Bundle ESs.


\section{Priority in Extended Bundle Event Structures}
\label{sec:PEBES}

The first extension of Bundle ESs we consider are \emph{Extended Bundle Event Structures (EBESs)}.
Bundle ESs were developed to give semantics to LOTOS in \cite{Langerak:Thesis}, but since the conflict relation was symmetric, they could not give semantics to the disable operator of LOTOS. Thus Extended Bundle ESs were introduced in the same reference.

Extended Bundle ESs are similar to Bundle ESs except that the conflict relation is replaced by the so-called \emph{asymmetric conflict} relation or \emph{disabling} relation. If an event $ e_1 $ disables another event $ e_2 $, denoted by $ e_2 \leadsto e_1 $, then once $ e_1 $ takes place $ e_2 $ cannot take place anymore, \ie $ e_1 $ can never precede $ e_2 $. Accordingly, disabling can be considered as an exclusion relation. Note that the asymmetric conflict or disabling relation is an irreflexive relation $ \leadsto \, \subseteq E \times E $ but is not necessarily asymmetric as the name suggests, \ie $ e_1 \leadsto e_2 \implies e_2 \not\leadsto e_1 $ does not necessarily hold for all events $ e_1 $ and $ e_2 $. Therefore Extended Bundle ESs are a generalization of Bundle ESs, and thus are more expressive \cite{Langerak:Thesis}.

Formally an Extended Bundle ES is a quadruple $ \varepsilon = \left( E, \leadsto, \mapsto, \labeling{} \right) $, where the stability condition is adapted as follows:
\vspace{-0.5em}
\begin{equation*}
	\label{eq:StabilityConstraintEBES}
	\text{\textbf{Stability:}} \quad \quad \forall X \subseteq E \logdot \forall e \in E \logdot X \mapsto e \implies  \left( \forall e_{1}, e_{2} \in X \logdot e_{1} \neq e_{2} \implies e_{1} \leadsto e_{2} \right) \tag{SC'}
	\vspace{-0.5em}
\end{equation*}
Note that stability again ensures that two distinct events of a bundle set are in mutual conflict. We adapt the Definitions of $ \en{\varepsilon}{\sigma} $, traces, $ \traces{\varepsilon} $, configurations, and $ \configurations{\varepsilon} $ accordingly.

Figure~\ref{fig:PEBESExample}~(a) shows an example of an EBES. The solid arrows denote causality, \ie reflect the enabling relation, where the bar between the arrows shows that they belong to the same bundle and the dashed line denotes again a mutual conflict as required by the stability condition \eqref{eq:StabilityConstraintEBES}. The dashed arrow denotes disabling, \eg $ b \leadsto a $ and $ ba \in \traces{\varepsilon} $ but $ ab \notin \traces{\varepsilon} $ in this example.

\begin{figure}[tb]
	\centering
	\vspace{-1em}
	\begin{tikzpicture}
		\node (e1)	at (0.5, 2)	{};
		\node (e2)	at (3, 2)	{};
		\node (e3)	at (0, 1)	{};
		\node (e4)	at (1, 1)	{};
		\node (e5)	at (2.5, 1)	{};
		\node (e6)	at (3.5, 1)	{};
		\node (e7)	at (0.5, 0)	{};
		\node (e8)	at (3, 0)	{};
		\node (a)	at (1.75, -1)	{(a)};
		\fill (e1) circle (2.5pt) node [above left]	{$ a $};
		\fill (e2) circle (2.5pt) node [above right]	{$ b $};
		\fill (e3) circle (2.5pt) node [above left]	{$ c $};
		\fill (e4) circle (2.5pt) node [above right]	{$ d $};
		\fill (e5) circle (2.5pt) node [above left]	{$ e $};
		\fill (e6) circle (2.5pt) node [above right]	{$ f $};
		\fill (e7) circle (2.5pt) node [below]		{$ g $};
		\fill (e8) circle (2.5pt) node [below]		{$ h $};
		\draw[enabling]			(e4) -- (e3);
		\draw\enabling{2}{2}		(e2) -- (e5);
		\draw\enabling{2}{2}		(e5) -- (e8);
		\draw\enabling{2}{2}		(e6) -- (e8);
		\draw[conflict]			(e5) -- (e6);
		\draw[thick] 			(2.75, 0.5) -- (3.25, 0.5);
		\draw[disabling]			(e2) -- (e1);
		\draw\disabling{2}{2}	(e1) -- (e4);
		\draw\disabling{2}{2}	(e3) -- (e7);
		\node (e1)	at (6.75, 2.6)	{};
		\node (e2)	at (5.3, 1.8)	{};
		\node (e3)	at (8.2, 1.8)	{};
		\node (e4)	at (5, 0.9)		{};
		\node (e5)	at (6.2, 0.9)	{};
		\node (e6)	at (7.3, 0.9)	{};
		\node (e7)	at (8.5, 0.9)	{};
		\node (e8)	at (6.75, 0)		{};
		\node (b)	at (6.75, -1)	{(b)};
		\fill (e1) circle (2.5pt) node [above right]	{$ a $};
		\fill (e2) circle (2.5pt) node [left]		{$ b $};
		\fill (e3) circle (2.5pt) node [right]		{$ c $};
		\fill (e4) circle (2.5pt) node [left]		{$ d $};
		\fill (e5) circle (2.5pt) node [right]		{$ e $};
		\fill (e6) circle (2.5pt) node [left]		{$ f $};
		\fill (e7) circle (2.5pt) node [right]		{$ g $};
		\fill (e8) circle (2.5pt) node [below]		{$ h $};
		\draw\hasseP{2.3}{2.3}	(e1) -- (e2);
		\draw\hasseP{2.3}{2.3} 	(e1) -- (e3);
		\draw\hasseP{2.7}{2.7} 	(e1) -- (e4);
		\draw\hasseP{1.6}{1.6} 	(e1) -- (e5);
		\draw\hasseP{1.6}{1.6} 	(e1) -- (e6);
		\draw\hasseP{2.7}{2.7} 	(e1) -- (e7);
		\draw\hasseP{1.5}{1.5} 	(e1) -- (e8);
		\draw\hasseP{1.6}{1.6}	(e2) -- (e4);
		\draw\hasseP{2.7}{2.7}	(e2) -- (e5);
		\draw\hasseP{2.3}{2.3}	(e2) -- (e8);
		\draw\hasseP{2.7}{2.7}	(e3) -- (e6);
		\draw\hasseP{1.6}{1.6}	(e3) -- (e7);
		\draw\hasseP{2.3}{2.3}	(e3) -- (e8);
		\draw\hasseP{2}{2}		(e4) -- (e8);
		\draw\hasseP{1.9}{1.9}	(e5) -- (e8);
		\draw\hasseP{1.9}{1.9}	(e6) -- (e8);
		\draw\hasseP{2}{2}		(e7) -- (e8);
		\node (f1)	at (11.75, 2.6)	{};
		\node (f2)	at (10.3, 1.8)	{};
		\node (f3)	at (13.2, 1.8)	{};
		\node (f4)	at (10, 0.9)		{};
		\node (f5)	at (11.2, 0.9)	{};
		\node (f6)	at (12.3, 0.9)	{};
		\node (f7)	at (13.5, 0.9)	{};
		\node (f8)	at (11.75, 0)	{};
		\node (c)	at (11.75, -1)	{(c)};
		\fill (f1) circle (2.5pt) node [above right]	{$ a $};
		\fill (f2) circle (2.5pt) node [left]		{$ b $};
		\fill (f3) circle (2.5pt) node [right]		{$ c $};
		\fill (f4) circle (2.5pt) node [left]		{$ d $};
		\fill (f5) circle (2.5pt) node [right]		{$ e $};
		\fill (f6) circle (2.5pt) node [left]		{$ f $};
		\fill (f7) circle (2.5pt) node [right]		{$ g $};
		\fill (f8) circle (2.5pt) node [below]		{$ h $};
		\draw\hasseP{2.3}{2.3}	(f1) -- (f2);
		\draw\hasseP{2.3}{2.3} 	(f1) -- (f3);
		\draw\hasseP{1.6}{1.6} 	(f1) -- (f5);
		\draw\hasseP{1.6}{1.6} 	(f1) -- (f6);
		\draw\hasseP{2.7}{2.7} 	(f1) -- (f7);
		\draw\hasseP{1.5}{1.5} 	(f1) -- (f8);
		\draw\hasseP{1.6}{1.6}	(f2) -- (f4);
		\draw\hasseP{2.3}{2.3}	(f2) -- (f8);
		\draw\hasseP{2.7}{2.7}	(f3) -- (f6);
		\draw\hasseP{2.3}{2.3}	(f3) -- (f8);
		\draw\hasseP{2}{2}		(f4) -- (f8);
		\draw\hasseP{2}{2}		(f7) -- (f8);
	\end{tikzpicture}
	\vspace{-1.2em}
	\caption{A PEBES $ \left( \varepsilon, \lessdot \right) $ with $ \epsilon $ in (a) and $ \lessdot $ in (b) as a Hasse diagram with transitivity exposed. (c)~shows $ \lessdot $ after dropping redundant priority pairs.}
	\label{fig:PEBESExample}
	\vspace{-1em}
\end{figure}
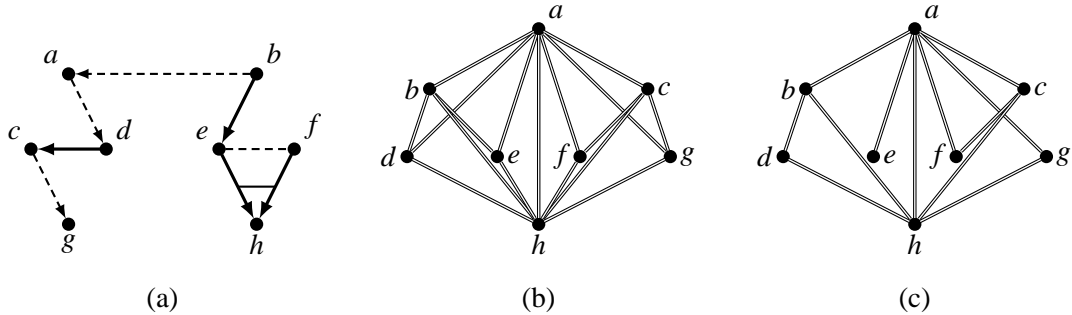

In order to define lposets of EBESs, we have to adapt the precedence relation $ \prec_{C} \; \subseteq C \times C $ such that it also covers disabling:
\vspace{-0.5em}
\begin{equation}
	\label{eq:lposetEBES}
	e \prec_{C} e' \iff \left( \exists X \subseteq E \logdot X \mapsto e' \land e \in X \right) \lor e \leadsto e'
	\vspace{-0.5em}
\end{equation}
Again $ \preceq_{C} $ denotes the reflexive and transitive closure of $ \prec_{C} $. The Definitions of lposets and $ \lposets{\varepsilon} $ are then adapted accordingly.
\cite{Langerak:Thesis} proves that $ \preceq_{C} $ is a partial order over $ C $ and that each linearization (obeying the defined precedence relations) of a given lposet built from an EBES yields an event trace of that structure. Furthermore, it is proved in \cite{Langerak:Thesis,Katoen:Thesis} that given two EBESs \(\varepsilon, \varepsilon'\) their lposets are equal iff their traces are equal, \ie $ \lposets{\varepsilon} = \lposets{\varepsilon'} \iff \traces{\varepsilon} = \traces{\varepsilon'} $.

$ \left( \varepsilon, \lessdot \right) = \left( E, \leadsto, \mapsto, \labeling{}, \lessdot \right) $ is a \emph{prioritized Extended Bundle ES (PEBES)}, where $ \varepsilon = \left( E, \leadsto, \mapsto, \labeling{} \right) $ is an EBES and $ \lessdot \subseteq \left( E \times E \right) $ is the acyclic priority relation.
Figure~\ref{fig:PEBESExample} illustrates an example of a PEBES with the EBES in Figure~\ref{fig:PEBESExample}~(a) and the priority relation in Figure~\ref{fig:PEBESExample}~(b).
A sequence of events $ \sigma = e_1 \dotsm e_n $ is a trace of $ \left( \varepsilon, \lessdot \right) $ iff
\begin{inparaenum}[1.)]
	\item $ \sigma \in \traces{\varepsilon} $ and
	\item $ \sigma $ satisfies the following constraint:
\end{inparaenum}
\vspace{-0.8em}
\begin{equation}
	\label{eq:PEBESTraceDef}
	\forall i<n \logdot \forall e_{j},e_{h} \in \bar{\sigma} \logdot e_{j} \neq e_{h} \wedge e_{j}, e_{h} \in \en{\varepsilon}{\sigma_{i}} \wedge e_{h} \lessdot e_{j} \implies j<h
	\vspace{-0.5em}
\end{equation}
$ C \in \configurations{ \varepsilon, \lessdot} $ iff $ \exists \sigma \in \traces{\varepsilon, \lessdot} \logdot \bar{\sigma} = C $.
Again $ \traces{ \varepsilon, \lessdot} \subseteq \traces{\varepsilon} $ and $ \lessdot' \subseteq \lessdot $ implies $ \traces{\varepsilon, \lessdot} \subseteq \traces{\varepsilon, \lessdot'} $.

\begin{lma}
	\label{lma:traceFilteringEBES}
	Let $ \left( \varepsilon, \lessdot \right) $ and $ \left( \varepsilon, \lessdot' \right) $ be two PEBES with $ \lessdot' \subseteq \lessdot $. Then $ \traces{\varepsilon, \lessdot} \subseteq \traces{\varepsilon, \lessdot'} $.
\end{lma}

\begin{proof}
	Straightforward from the Definition of traces, \eqref{eq:PEBESTraceDef}, and $ \lessdot' \subseteq \lessdot $.
%
%
\end{proof}

Similar to PBESs, we can remove a priority pair $ e \lessdot e' $ or $ e' \lessdot e $ between an event $ e $ and its cause $ e' $, because an event and its cause are never enabled together. Therefore \eg the pair $ e \lessdot b $ in Figure~\ref{fig:PEBESExample} is redundant because of $ \Set{ b } \mapsto e $. Also a priority pair $ e \lessdot e' $ between an event $ e' $ and its disabler $ e $, \ie for $ e' \leadsto e $, does not affect the semantics, since $ e $ must follow $ e' $ anyway. Consider for example the events $ a $ and $ d $ in Figure~\ref{fig:PEBESExample}. Because of $ a \leadsto d $, $ a $ always pre-empts $ d $ and thus $ d \lessdot a $ is redundant.

\begin{thm}
	\label{thm:PriorityReductionPEBES}
	Let $ \left( \varepsilon, \lessdot \right) $ be a PEBES and
	\vspace{-0.6em}
	\begin{align*}
		\lessdot' \deff \lessdot \setminus \Set{ \left( e, e' \right) \mid e' \leadsto e \lor \left( \exists X \subseteq E \logdot \left( e \in X \land X \mapsto e' \right) \lor \left( e' \in X \land X \mapsto e \right) \right) }.
	\end{align*}
	\vspace{-2em}\\
	Then $ \traces{\varepsilon, \lessdot} = \traces{\varepsilon, \lessdot'} $.
\end{thm}

\begin{proof}
	Similar to the proof of Theorem~\ref{thm:PriorityReduction}, where the second case is replaced by:
	\vspace{-0.3em}
	\begin{description}
		\item[Case $ e_j \leadsto e_h $:] Because of the Definition of traces, $ \sigma = e_1 \dotsm e_n \in \traces{\varepsilon} $ and $ e_j, e_h \in \bar{\sigma} $ imply that $ e_j \in \en{\varepsilon}{\sigma_{j - 1}} $. Then $ e_j \leadsto e_h \wedge e_j \in \en{\varepsilon}{\sigma_{j - 1}} \wedge e_h \in \bar{\sigma} \implies j < h $.
	\end{description}
	\vspace{-1.7em}
\end{proof}

\noindent
Note that priority for events that are directly related by the bundle enabling relation is always redundant, regardless whether the cause has the higher priority or the effect does. On the other hand we can reduce pairs of events that are related by disabling only if the event has a higher priority than its disabler. Consider for example the PEBES $ \left( \Set{ e, e' }, \Set{ e \leadsto e' }, \emptyset, \labeling{}, \Set{ e \lessdot e' } \right) $. The only traces of this PEBES are $ e $ and $ e' $, but if we remove the priority pair $ e \lessdot e' $ we have the additional trace $ e e' $. Similarly we cannot remove the $ e \leadsto e' $ here, because this yields to the additional trace $ e' e $.

The result of dropping redundant priority pairs for the PEBES in Figure~\ref{fig:PEBESExample} as described by Theorem~\ref{thm:PriorityReductionPEBES} is presented in Figure~\ref{fig:PEBESExample}~(c). Note that after dropping redundant pairs the priority relation is not a partial order anymore.

Again limiting our attention to a specific configuration allows us to ignore some more priority pairs.
In contrast to PBESs we can sometimes also ignore priority pairs that overlap with disabling. Consider for example $ b \leadsto a $ and $ a \leadsto d $ of the PEBES in Figure~\ref{fig:PEBESExample}. The priority pair $ d \lessdot b $ is redundant with respect to the configuration $ \Set{ a, b, d } $ but not with respect to the configuration $ \Set{ b, d } $.
Note that again the direction of the priority pair is important in the case of indirect disabling but not in the case of indirect enabling. If we for instance replace $ d \lessdot b $ in Figure~\ref{fig:PEBESExample} by $ b \lessdot d $, then $ \Set{ a, b, d } $ is not a configuration anymore and $ b \lessdot d $ is not redundant in all remaining configurations containing $ b $ and $ d $.

The cases in which priority pairs are redundant with respect to some configuration $ C $ are again well described by the precedence relation $ \preceq_C $, \ie we can identify redundant priority pairs easily from the lposet of $ C $.
The priority pair $ h \lessdot b $ is obviously redundant in the case of
\begin{tikzpicture}
	\node[lposet, minimum height=4mm, minimum width=4mm] (l1) at (0, 0) {
		\begin{tikzpicture}
			\node (e2) at (0, 0.04)		{$ b $};
			\node (e5) at (0.8, 0)		{$ e $};
			\node (e8) at (1.6, 0.04)	{$ h $};
			\draw[->, thick] (0.2, 0) -- (e5);
			\draw[->, thick] (e5) -- (1.4, 0);
		\end{tikzpicture}
	};
\end{tikzpicture}
but not in the case of
\begin{tikzpicture}
	\node[lposet, minimum height=4mm, minimum width=4mm] (l1) at (0, 0) {
		\begin{tikzpicture}
			\node (e2) at (0, 0.01)		{$ b $};
			\node (e6) at (0.4, 0)		{$ f $};
			\node (e8) at (1.2, 0.01)	{$ h $};
			\draw[->, thick] (e6) -- (1, 0);
		\end{tikzpicture}
	};
\end{tikzpicture}
and $ d \lessdot b $ is obviously redundant in the case of
\begin{tikzpicture}
	\node[lposet, minimum height=4mm, minimum width=4mm] (l1) at (0, 0) {
		\begin{tikzpicture}
			\node (e2) at (0, 0.04)		{$ b $};
			\node (e1) at (0.8, 0)		{$ a $};
			\node (e4) at (1.6, 0.04)	{$ d $};
			\draw[->, thick] (0.2, 0) -- (e1);
			\draw[->, thick] (e1) -- (1.3, 0);
		\end{tikzpicture}
	};
\end{tikzpicture}
but not in the case of
\begin{tikzpicture}
	\node[lposet, minimum height=4mm, minimum width=4mm] (l1) at (0, 0) {
		\begin{tikzpicture}
			\node (e2) at (0, 0)		{$ b $};
			\node (e4) at (0.4, 0)	{$ d $};
		\end{tikzpicture}
	};
\end{tikzpicture}.
Let $ \reducedto{\traces{\varepsilon, \lessdot}}{C} \deff \Set{ \sigma \mid \sigma \in \traces{\varepsilon, \lessdot} \wedge \bar{\sigma} = C } $ be the set of traces over $ C $. Then for all configurations $ C $ we can ignore all priority pairs $ e \lessdot e' $ such that $ e' \preceq_C e $.

\begin{thm}
	\label{thm:PriorityIgnorancePEBES}
	Let $ \left( \varepsilon, \lessdot \right) $ be a PEBES, $ \left\langle C, \preceq_C, \labeling{} \right\rangle \in \lposets{\varepsilon} $, and $ \lessdot' \deff \lessdot \setminus \Set{ \left( e,e' \right) \in C \times C \mid e' \preceq_C e } $. Then:
	\vspace{-1em}
	\begin{align*}
		\reducedto{\traces{\varepsilon, \lessdot}}{C} \; = \reducedto{\traces{\varepsilon, \lessdot'}}{C}
	\end{align*}
\end{thm}

\begin{proof}
	Note that $ \reducedto{\traces{\varepsilon, \lessdot}}{C} \; \subseteq \traces{\varepsilon, \lessdot} $ and $ \reducedto{\traces{\varepsilon, \lessdot'}}{C} \; \subseteq \traces{\varepsilon, \lessdot'} $.
	
	By Lemma~\ref{lma:traceFilteringEBES}, $ \traces{\varepsilon, \lessdot} \subseteq \traces{\varepsilon, \lessdot'} $ and thus also $ \reducedto{\traces{\varepsilon, \lessdot}}{C} \; \subseteq \reducedto{\traces{\varepsilon, \lessdot'}}{C} $.
	
	To show $ \reducedto{\traces{\varepsilon, \lessdot'}}{C} \; \subseteq \reducedto{\traces{\varepsilon, \lessdot}}{C} $, assume a trace $ \sigma = e_1 \dotsm e_n \in \reducedto{\traces{\varepsilon, \lessdot'}}{C} $. We have to show that $ \sigma \in \reducedto{\traces{\varepsilon, \lessdot}}{C} $, \ie that $ \forall e \in \bar{\sigma} \logdot e \in C $, $ \sigma \in \traces{\varepsilon} $, and that $ \sigma $ satisfies Condition~\eqref{eq:PEBESTraceDef}. $ \forall e \in \bar{\sigma} \logdot e \in C $ and $ \sigma \in \traces{\varepsilon} $ follows from $ \sigma \in \reducedto{\traces{\varepsilon, \lessdot'}}{C} $ by the Definition of traces of PEBESs. $ \sigma $ satisfies Condition~\eqref{eq:PEBESTraceDef} if $ \forall i < n \logdot \forall e_j, e_h \in \bar{\sigma} \logdot e_j \neq e_h \wedge e_j, e_h \in \en{\varepsilon}{\sigma_i} \wedge e_h \lessdot e_j \implies j < h $. Let us fix $ i < n $ and $ e_j, e_h \in \bar{\sigma} $. Assume $ e_j \neq e_h $, $ e_j, e_h \in \en{\varepsilon}{\sigma_i} $, and $ e_h \lessdot e_j $. It remains to prove that $ j < h $.
	
	Because of the Definition of $ \lessdot' $, assumption $ e_h \lessdot e_j $ implies that $ e_h \lessdot' e_j $ or $ e_j \preceq_C e_h $. In the first case $ j < h $ follows, because of the Definition of traces and \eqref{eq:PEBESTraceDef}, from $ e_h \lessdot' e_j $, $ e_j, e_h \in \bar{\sigma} $, and $ \sigma \in \reducedto{\traces{\varepsilon, \lessdot'}}{C} $. The other case, \ie that $ e_j \preceq_C e_h $ and $ e_j \neq e_h $ implies $ j < h $, was already proved in \cite{Langerak:Thesis}.
\end{proof}

\noindent
Consider once more the PEBES $ \left( \varepsilon, \lessdot \right) $ of Figure~\ref{fig:PEBESExample} with respect to the configuration $ \Set{ b, e, h } $. We have $ \Set{ b } \mapsto e $, $ \Set{ e, f } \mapsto h $, and $ h \lessdot b $. As explained before we cannot drop the priority pair $ h \lessdot b $, because of the trace $ fhb \notin \traces{\varepsilon, \lessdot} $. However with Theorem~\ref{thm:PriorityIgnorancePEBES} we can ignore $ h \lessdot b $|and also $ h \lessdot e $ and $ e \lessdot b $|for the semantics of $ \left( \varepsilon, \lessdot \right) $ if we limit our attention to $ \Set{ b, e, h } $, because $ \reducedto{\traces{\varepsilon, \lessdot}}{\Set{ b, e, h }} = \reducedto{\traces{\varepsilon, \left( \lessdot \setminus \Set{ h \lessdot b, h \lessdot e, e \lessdot b } \right)}}{\Set{ b, e, h }} $.

Similarly we can ignore $ c \lessdot a $ if we limit our attention to the configuration $ C = \Set{ a, c, d } $, since $ \reducedto{\traces{\varepsilon, \lessdot}}{C} = \reducedto{\traces{\varepsilon, \left( \lessdot \setminus \Set{ c \lessdot a } \right)}}{C} $.
Note that here the precedence pair $ a \preceq_{C} c $ that allows us to ignore $ c \lessdot a $ results from the correlation between a disabling pair $ a \leadsto d $ and an enabling pair $ \Set{ d } \mapsto c $. Thus with Theorem~\ref{thm:PriorityIgnorancePEBES} we can ignore even priority pairs that are redundant in specific situations because of combining enabling and disabling.

This combination prohibits us on the other hand from ignoring priority of the opposite direction, the direction which is compatible with the precedence direction. That is possible only with precedence resulted from enabling purely as it is the case in Theorem~\ref{thm:PriorityIgnoranceBES} for PBESs.
For instance, suppose that $ b \lessdot h $ for the structure in Figure~\ref{fig:PEBESExample} then we can ignore this priority pair in a configuration $ \Set{ b, e, h } $. That is not formulated in the Theorem~\ref{thm:PriorityIgnorancePEBES} above, since $ \preceq_C $ abstracts from the relation between events. While in contrast to EBESs, the conflict relation is symmetric in Bundle ESs, and precedence results only from enabling. Thus, in contrast to PBESs, we do not have minimality of priority ignorance in PEBESs.


\section{Priority in Dual Event Structures}
\label{sec:PDES}

Dual ESs are the second extension of BES examined here. The stability constraint in BESs and EBESs prohibits two events from the same bundle to take place in the same system run. Thus \eg $ \Set{ b, e, f, h } $ is not a configuration of the EBES in Figure~\ref{fig:PEBESExample}. It provides some kind of stability to the causality in the structure. More precisely due to stability in every trace or lposet for every event the necessary causes can be determined.
Without the conflict between $ e $ and $ f $ in the example the trace $ befh $ is possible. But then it is impossible to determine whether $ h $ is caused in $ befh $ by $ e $ or $ f $.
The stability constraint prohibits such ambiguity. On the other hand such a constraint limits the expressiveness of the structure, and forces an XOR condition between the elements of a bundle set. In some system specifications a more relaxed definition may be useful. Dual ESs provide such a relaxed definition.

The definition of Dual ESs varies between \cite{Katoen:Thesis} and \cite{Langerak97causalambiguity}, but both show the causal ambiguity explained above. In \cite{Katoen:Thesis} Dual ESs are based on Extended Bundle ESs, while in \cite{Langerak97causalambiguity} they are based on Bundle ESs. Since we have studied the relation between disabling and priority, and want to focus here on the effect of causal ambiguity on priority, we analyze the version of \cite{Langerak97causalambiguity}. A \emph{Dual Event Structure (DES)} is a quadruple $ \Delta = \left( E, \#, \mapsto, \labeling{} \right) $ similar to Definition~\ref{def:BES} but without the stability condition.

The Definitions of $ \en{\Delta}{\sigma} $, traces, and $ \traces{\Delta} $ are similar to Section~\ref{sec:PBES} for BESs. Of course the deletion of the stability condition leads to additional traces, \eg for structure in Figure~\ref{fig:PBESExample}~(a) we obtain the additional traces $ abcd, abdc, acbd, acdb, cabd, cadb $, and $ cdab $.
Figure~\ref{fig:ExamplaryDualEventStructure} shows a DES taken from \cite{Langerak97causalambiguity}.

\begin{figure}[t]
	\centering
	\begin{tikzpicture}
		\node (a1) at (0, 1)			{};
		\node (b1) at (0.8, 1)		{};
		\node (c1) at (1.6, 1)		{};
		\node (d1) at (0.8, 0)		{};
		\node (a)  at (0.8, -0.8)	{(a)};
		\fill (a1) circle (2.5pt) node [above]	{$a$};
		\fill (b1) circle (2.5pt) node [above]	{$b$};
		\fill (c1) circle (2.5pt) node [above]	{$c$};
		\fill (d1) circle (2.5pt) node [below]	{$d$};
		\draw\enabling{2.4}{2.4}	(a1) -- (d1);
		\draw[enabling]			(b1) -- (d1);
		\draw\enabling{2.4}{2.4}	(c1) -- (d1);
		\draw[thick]				(0.5, 0.37) -- (0.8, 0.5);
		\draw[thick]				(0.8, 0.5) -- (1.1, 0.37);
		\node (a2) at (3, 1)		{};
		\node (b2) at (3.8, 1)	{};
		\node (c2) at (4.6, 1)	{};
		\node (d2) at (3.8, 0)	{};
		\node (b)  at (3.8, -0.8)	{(b)};
		\fill (a2) circle (2.5pt) node [above]	{$a$};
		\fill (b2) circle (2.5pt) node [above]	{$b$};
		\fill (c2) circle (2.5pt) node [above]	{$c$};
		\fill (d2) circle (2.5pt) node [below]	{$d$};
		\draw\enabling{2.4}{2.4}	(a2) -- (d2);
		\draw[enabling]			(b2) -- (d2);
		\draw\enabling{2.4}{2.4}	(c2) -- (d2);
		\draw[thick]				(3.5, 0.37) -- (3.8, 0.5);
		\draw[thick]				(3.8, 0.5) -- (4.1, 0.37);
		\draw[priority]			(d2) to [bend right] (c2);
	\end{tikzpicture}
	\vspace{-1em}
	\caption{A Dual ES without priority in (a) and with priority in (b).}
	\label{fig:ExamplaryDualEventStructure}
	\vspace{-1em}
\end{figure}
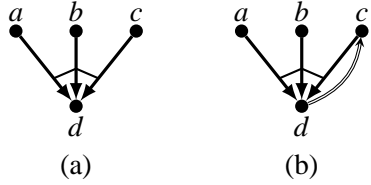

Causal ambiguity affects the way lposets are built, since the causal order is not clear\footnote{Since there is no disable relation here, causality is the only source of order in lposets.}. In \cite{Langerak97causalambiguity} Langerak et al.\ tried to solve this problem. They illustrated that there are different causality interpretations possible for causal ambiguity. They defined five different intensional posets: liberal, bundle-satisfaction, minimal, early and late posets. Intensional means posets are defined depending on the causality relations in the structure, while observational on the other hand means posets are obtained out of event traces (where no structure of the system is available, but only behavior). We examine these different kinds and their relations informally for brevity.

Langerak et al.\ illustrated that in order to detect the cause of an event like \(d\) in the trace $abcd$ of Figure~\ref{fig:ExamplaryDualEventStructure}, one should consider the prefix of \(d\) (i.e. \(abc\)) and then can have the following interpretations:
\begin{compactitem}
	\item Liberal Causality: means that any cause out of the prefix is accepted as long as bundles are satisfied: \(abc, ab, b, ac, bc\), \etc are all accepted as a cause for \(d\).
Then all events in a cause precede $d$ in the built poset, \eg the posets for the last causes are
		\begin{tikzpicture}
			\node[lposet, minimum height=4mm, minimum width=4mm] (l1) at (0, 0) {
				\begin{tikzpicture}
					\node (a) at (0, 0.5)		{$ a $};
					\node (b) at (0, 0.25)		{$ b $};
					\node (c) at (0, 0)			{$ c $};
					\node (d) at (0.75, 0.25)	{$ d $};
					\draw[->, thick] (a) -- (d);
					\draw[->, thick] (b) -- (d);
					\draw[->, thick] (c) -- (d);
				\end{tikzpicture}
			};
		\end{tikzpicture},
		\begin{tikzpicture}
			\node[lposet, minimum height=4mm, minimum width=4mm] (l1) at (0, 0) {
				\begin{tikzpicture}
					\node (a) at (0, 0.5)		{$ a $};
					\node (b) at (0, 0.25)		{$ b $};
					\node (c) at (0, 0)			{$ c $};
					\node (d) at (0.75, 0.25)	{$ d $};
					\draw[->, thick] (a) -- (d);
					\draw[->, thick] (b) -- (d);
				\end{tikzpicture}
			};
		\end{tikzpicture},
		\begin{tikzpicture}
			\node[lposet, minimum height=4mm, minimum width=4mm] (l1) at (0, 0) {
				\begin{tikzpicture}
					\node (a) at (0, 0.5)		{$ a $};
					\node (b) at (0, 0.25)		{$ b $};
					\node (c) at (0, 0)			{$ c $};
					\node (d) at (0.75, 0.25)	{$ d $};
					\draw[->, thick] (b) -- (d);
				\end{tikzpicture}
			};
		\end{tikzpicture},
		\begin{tikzpicture}
			\node[lposet, minimum height=4mm, minimum width=4mm] (l1) at (0, 0) {
				\begin{tikzpicture}
					\node (a) at (0, 0.5)		{$ a $};
					\node (b) at (0, 0.25)		{$ b $};
					\node (c) at (0, 0)			{$ c $};
					\node (d) at (0.75, 0.25)	{$ d $};
					\draw[->, thick] (a) -- (d);
					\draw[->, thick] (c) -- (d);
				\end{tikzpicture}
			};
		\end{tikzpicture},
		\begin{tikzpicture}
			\node[lposet, minimum height=4mm, minimum width=4mm] (l1) at (0, 0) {
				\begin{tikzpicture}
					\node (a) at (0, 0.5)		{$ a $};
					\node (b) at (0, 0.25)		{$ b $};
					\node (c) at (0, 0)			{$ c $};
					\node (d) at (0.75, 0.25)	{$ d $};
					\draw[->, thick] (b) -- (d);
					\draw[->, thick] (c) -- (d);
				\end{tikzpicture}
			};
		\end{tikzpicture},
		\etc, respectively. We use the same mechanism of building posets for the next types of causalities.
	\item Bundle-Satisfaction Causality: bundles are satisfied by exactly one event: \(b, ab, ac\) are accepted causes but not \(abc\).
	\item Minimal Causality: bundles are satisfied so that no subset of a cause is accepted. So \(b, ac\) are accepted but not \(ab\) or \(bc\).
	\item Early Causality: the earliest bundle-satisfaction cause is accepted: \(b\) is accepted, but not \(ac\) as \(c\) happened later than \(b\).\footnote{In \cite{AmbiguityReport} a procedure is defined to detect how early is a cause depending on binary numbers.}
	\item Late Causality: (cf. \cite{Langerak97causalambiguity}, it will be skipped here).
\end{compactitem}

\cite{Langerak97causalambiguity} shows that equivalence in one kind of posets between two structures implies equivalence in some other specific kinds of posets, but equivalence in any of the kinds implies equivalence in traces.

We add priority to DESs in the same way as before. A \emph{prioritized Dual ES (PDES)} is the tuple $ \left( \Delta, \lessdot \right) $, where $ \Delta $ is a DES and $ \lessdot $ is the acyclic priority relation. Also the definitions of traces, $ \traces{\Delta, \lessdot} $, configurations, and $ \configurations{\Delta, \lessdot} $ are adapted similar as in the section before.

Since the conflict relation provided here is the same as in Bundle ESs, we can remove redundant priority pairs that overlap with the conflict relation as described in Theorem~\ref{thm:PriorityReduction}, \ie whenever there is $ e \# e' $ or $ e' \# e $ then $ e \lessdot e' $ is redundant and can be removed. The situation for enabling is different because of the missing stability condition.
The priority pair $ c \lessdot d $ in the PDES in Figure~\ref{fig:ExamplaryDualEventStructure} is not redundant, because it removes some traces. The reason is that $ c $ is not anymore a necessary cause for $ d $, since $ d $ can be enabled by $ b $ even if $ c $ occurs in the same trace. So at the structure level of PDESs we cannot in general remove priority pairs because of overlappings with the enabling relation.

In the case of PBESs and PEBESs partial orders help to identify redundant priority pairs at the configuration level.
Unfortunately, we cannot do the same here.
Let us consider the configuration $ \Set{ a, b, c, d } $, and consider liberal causality. Indeed applying \eqref{thm:PriorityIgnoranceBES} on the poset
\begin{tikzpicture}
	\node[lposet, minimum height=4mm, minimum width=4mm] (l1) at (0, 0) {
		\begin{tikzpicture}
			\node (a) at (0, 0.3)	{$ a $};
			\node (b) at (0, 0.04)	{$ b $};
			\node (c) at (1.6, 0.15)	{$ c $};
			\node (d) at (0.8, 0.19)	{$ d $};
			\draw[->, thick] (a) -- (0.6, 0.2);
			\draw[->, thick] (0.2, 0) -- (0.6, 0.1);
			\draw[->, thick] (c) -- (1, 0.15);
		\end{tikzpicture}
	};
\end{tikzpicture}
and the priority $c \lessdot d$ yields the sequence \(abcd\) which is not a trace. On the other hand, considering bundle-satisfaction causality, the poset
\begin{tikzpicture}
	\node[lposet, minimum height=4mm, minimum width=4mm] (l1) at (0, 0) {
		\begin{tikzpicture}
			\node (a) at (0, 0.125)		{$ a $};
			\node (b) at (0.4, 0.29)		{$ b $};
			\node (c) at (0.4, 0)		{$ c $};
			\node (d) at (1.2, 0.165)	{$ d $};
			\draw[->, thick] (0.6, 0.25) -- (1, 0.175);
			\draw[->, thick] (c) -- (1, 0.075);
		\end{tikzpicture}
	};
\end{tikzpicture}
with the same priority yields the same sequence \(abcd\) again. The same will be for minimal causality and a poset like
\begin{tikzpicture}
	\node[lposet, minimum height=4mm, minimum width=4mm] (l1) at (0, 0) {
		\begin{tikzpicture}
			\node (a) at (0, 0.25)		{$ a $};
			\node (b) at (1.2, 0.165)	{$ b $};
			\node (c) at (0, 0)			{$ c $};
			\node (d) at (0.8, 0.165)	{$ d $};
			\draw[->, thick] (a) -- (0.6, 0.175);
			\draw[->, thick] (c) -- (0.6, 0.075);
		\end{tikzpicture}
	};
\end{tikzpicture}.

In fact none of the mentioned kinds of posets can be used alone without the priority, and thus ignorance is not possible with causal ambiguity \wrt a single poset for a configuration. Even when \(c \lessdot d \) seems to yield pre-emption in the poset
\begin{tikzpicture}
	\node[lposet, minimum height=4mm, minimum width=4mm] (l1) at (0, 0) {
		\begin{tikzpicture}
			\node (a) at (0, 0)			{$ a $};
			\node (b) at (0.4, 0.04)		{$ b $};
			\node (c) at (1.6, 0)		{$ c $};
			\node (d) at (1.2, 0.04)	{$ d $};
			\draw[->, thick] (b) -- (d);
		\end{tikzpicture}
	};
\end{tikzpicture},
one can have the linearization \(acbd\) which is a trace, and priority turns out to be redundant in this very trace (but not in the whole poset). The reason is partial orders here do not necessarily represent causes.

\section{Conclusions}
\label{sec:conclusions}

We have added priority to different Event Structures: Prime ESs as a simple model with conjunctive causality, Bundle ESs with disjunctive causality, Extended Bundle ESs with asymmetric conflict, and Dual ESs with causal ambiguity. In all cases, priority leads to trace filtering and limits concurrency and non-determinism.
We then analyze the relationship between the new priority relation and the other relations of the ESs. Since priority has an effect only if the related events are enabled together, overlappings between the priority relation and the other relations of the ESs sometimes lead to redundant priority pairs.

In PPESs, PBESs, and PEBESs priority is redundant between events that are related directly by causality. Moreover in all considered ESs priority is redundant between events that are related directly by the conflict relation. But in the case of PEBESs the conflict relation implements asymmetric conflicts. Hence in contrast to the other ESs we have to take the direction of the disabling relation into account.

The main difference between redundancy of priority in PPESs and the other three models is due to events that are indirectly related by causality. In PPESs causality is a transitive relation, \ie all pairs which are indirectly related by causality are directly related by causality as well. The enabling relation of the other models is not transitive. Thus priority pairs between events that are only indirectly related by enabling are not necessarily redundant. Unfortunately and unlike PPESs, this means that we cannot ensure after removing the redundant priority pairs that the remaining priority pairs necessarily lead to pre-emption. So the other models hold more ambiguity to a modeler.

Instead we show that if we limit our attention to a specific configuration $ C $ a priority pair $ e \lessdot e' $ is redundant if $ e' \preceq_C e \vee e \preceq_C e' $ for PBESs or if $ e' \preceq_C e $ for PEBESs. This allows us to ignore, for the semantics with respect to specific configurations, additionally priority pairs between events indirectly related by enabling for PBESs and by enabling, by disabling, or even by combinations of enabling and disabling for PEBESs. In the case of PBESs we obtain a minimality result this way.

Unfortunately in PDESs even priority pairs between events that are directly related by causality are not necessarily redundant.
So from a modeler’s perspective, priority in DESs hold the biggest ambiguity among all the studied ESs. In other words, one cannot figure out the role priority plays at design time or structure level, and whether this priority yields pre-emption or not. Even at the configuration level, that is not possible in general due to causal ambiguity.

Thus our main contributions are:
\begin{inparaenum}[1.)]
	\item We add priority as a binary acyclic relation on events to ESs.
	\item We show that the relation between priority and other event relations of an ES can lead to redundant priority pairs, \ie to priority pairs that do never (or at least for some configurations not) affect the behavior of the ES.
	\item Then we show how to completely remove such pairs in PPESs and that this is in general not possible in ESs with a more complex causality model like PBESs, PEBESs, or PDESs.
	\item Instead we show how to identify all priority pairs that are redundant with respect to configurations in PBESs and that the situation in PEBESs and DESs is different.
	\item We show how to identify (some of the) redundant priority pairs at the level of configurations in PEBESs and
	\item that again this is in general not possible in the same way for PDESs.
\end{inparaenum}

After dropping or ignoring redundant priority pairs as described above, the minimum potential for overlapping between priority and causality can be found in PPESs, while the maximum is in PDESs. In PPESs all remaining priority pairs indeed affect the semantics, \ie exclude traces. In PBESs the same holds with respect to specific configurations.
In PEBESs after dropping the redundant priority pairs the disabling relation has no overlapping with only priority directed in the opposite direction.

In Section~\ref{sec:PriorityBES} we show that adding priority complicates the definition of families of lposets to capture the semantics of prioritized ESs. We observe that because of priority a single configuration may require several lposets to describe its semantics. The same already applies for DESs because of the causal ambiguity. However note that priority does not lead to causal ambiguity. Thus, we can define the semantics of prioritized ESs by families of lposets if we do not insist on the requirement that there is exactly one lposet for each configuration. We leave the definition of such families of lposets for future work.
Such families of lposets for prioritized ESs may also help to identify and ignore redundant priority pairs in the case of PEBESs and PDESs.
Another interesting topic for further research is to analyze how priority influences the expressive power of ESs.

\addcontentsline{toc}{section}{References}
\bibliographystyle{eptcs}
\bibliography{priority.bib}

\end{document}